\documentclass[11pt,centertags,reqno,twoside]{amsart}
\usepackage{amssymb}
\usepackage{mathptmx}

\DeclareSymbolFont{SY}{U}{psy}{m}{n}
\DeclareMathSymbol{\emptyset}{\mathord}{SY}{'306}

\renewcommand{\eqref}[1]{{\rm(\ref{#1})}}

\newcommand{\bbC}{{\mathbb C}}
\newcommand{\bbR}{{\mathbb R}}

\newcommand{\bbN}{{\mathbb N}}

\newcommand{\re}{{\rm e}}
\newcommand{\ri}{{\rm i}}

\newcommand{\sE}{{\sf E}}

\newcommand{\fH}{\mathfrak{H}}

\newcommand{\fM}{\mathfrak{M}}
\newcommand{\fN}{\mathfrak{N}}
\newcommand{\fP}{\mathfrak{P}}
\newcommand{\fQ}{\mathfrak{Q}}
\newcommand{\fR}{\mathfrak{R}}

\newcommand{\lal}{{\langle}}
\newcommand{\ral}{{\rangle}}

\newcommand{\be}{\begin{equation}}
\newcommand{\ee}{\end{equation}}

\DeclareMathOperator{\spec}{spec}

\newcommand{\Ran}{\mathop{\mathrm{Ran}}}

\newcommand{\Dom}{\mathop{\mathrm{Dom}}}

\allowdisplaybreaks

\numberwithin{equation}{section}

\newtheorem{theorem}{Theorem}[section]
\newtheorem{corollary}[theorem]{Corollary}
\newtheorem{lemma}[theorem]{Lemma}
\newtheorem{proposition}[theorem]{Proposition}

\theoremstyle{definition}

\theoremstyle{remark}
{\it}{\rm}
\newtheorem{remark}[theorem]{Remark}
\newtheorem{example}[theorem]{Example}

\begin{document}

\title[Optimal bounds on the speed of subspace evolution]
{
Optimal bounds on the speed of subspace evolution$^*$$^\dagger$}\thanks{$^*$%
This work was partially supported by the Heisenberg-Landau Program.}%
\thanks{$^\dagger$Published in \textit{J. Phys. A: Math. Theor.}
\textbf{55}:23 (2022), 235203 [17 pp.] (https://doi.org/10.1088/1751-8121/ac6bcf).}

\author[S. Albeverio and A. K. Motovilov]
{Sergio Albeverio and  Alexander K. Motovilov}

\address{Sergio Albeverio,
Institute of Applied Mathematics, and Hausdorff Center of
Mathematics, University of Bonn, Endenicher Allee 60, D-53115 Bonn,
Germany} \email{albeverio@uni-bonn.de}

\address{Alexander K. Motovilov, Bogoliubov Laboratory of
Theoretical Physics, JINR, Joliot-Cu\-rie 6, 141980 Dubna, Russia,
and Dubna State University, Universitetskaya 19,
141980 Dubna, Russia} \email{motovilv@theor.jinr.ru}




\begin{abstract}

By a quantum speed limit one usually understands an estimate on how
fast a quantum system can evolve between two distinguishable states.
The most known quantum speed limit is given in the form of the
celebrated Mandelstam-Tamm inequality that bounds the speed of the
evolution of a state in terms of its energy dispersion. In contrast
to the basic Mandelstam-Tamm inequality, we are concerned not with a
single state but with a (possibly infinite-dimensional) subspace
which is subject to the Schr\"odinger evolution. By using the
concept of maximal angle between subspaces we derive optimal bounds
on the speed of such a subspace evolution. These bounds may be
viewed as further generalizations of the Mandelstam-Tamm inequality.
Our study includes the case of unbounded Hamiltonians.

\end{abstract}

\maketitle

\section{Introduction}
\label{SIntro}

By a \textit{quantum speed limit} one usually calls a lower bound on
the time that is needed for a quantum system to evolve from a given
state to a target state or a target subspace. The history of the
subject is already long, being traced back to the 1945's pioneering
work by Mandelstam and Tamm \cite{MaT}. The volume of literature on
quantum speed limits and their applications in a variety of areas is
large and by no means we make here an attempt to present a more
or less complete review of all relevant results. Instead, we only
inform the interested reader that comprehensive surveys of the
literature on various quantum speed limits may be found in the
recent review articles \cite{DeCa} and \cite{Frey} (see also the
introductory part of the very recent paper \cite{Lychko}).

We begin with recalling how the main known quantum speed
limits look. To this end, we consider an isolated quantum system
described by a Hamiltonian $H$, which is assumed to be a
time-independent self-adjoint operator acting in the complex Hilbert
space $\mathfrak{H}$. Any vector $\phi$ from the unit sphere in
$\fH$ represents a possible pure state of this system. Strictly
speaking, a pure state $\mathcal S$ is rather a class of equivalence
of norm-one vectors in $\mathfrak{H}$: the vectors
$\phi,\psi\in\mathfrak{H}$ with $\|\phi\|=\|\psi\|=1$ represent the
same pure state if  $\psi=u\phi$ for some $u\in\bbC$ such that
$|u|=1$. In an obvious way, the state $\mathcal S$ may be identified
with a one-dimensional subspace $\mathfrak{P}_{\mathcal{S}}$
which is the span of an arbitrarily chosen vector $\psi$ in
$\mathcal S$,\,\, $\mathfrak{P}_{\mathcal S}:=\bigl\{f=\lambda
\psi\,\,\, \big| \lambda\in\mathbb{C}\,\bigr\}$.

In what follows we will always suppose that units of measurement are
chosen such that $\hbar=1$. The time evolution of a state
vector $\psi(t)\in\mathfrak{H}$, $t\in\mathbb{R}$, is assumed to be
governed by the Schr\"odinger equation
\begin{align}
\label{Sch1p}
\mathrm{i} \frac{d}{dt}\psi & =H\psi,\\
\label{Sch2p} \psi(t)\bigr|_{t=0}&=\psi_0,
\end{align}
where the initial-state vector $\psi_0$, $\|\psi_0\|=1$, along with all other
vectors on the path $\psi(t)$, $t\in\bbR$, should belong to
the domain $\mathrm{Dom}(H)$ of the Hamiltoninan $H$.

Starting point in the study of quantum speed limits was the
following natural question: \textrm{How fast can a quantum system
with the Hamiltonian $H$ arrive at a state orthogonal to its initial
state~$\psi_0$? }

It is obvious that the answer to this question may be important in various
respects. Perhaps, the very latest motivation stems from quantum information
theory and quantum computing (see, e.g., \mbox{\cite{DeCa,Frey}}).

Available answers to the above question have been given in the form
of lower bounds for the orthogonalization time $T_\perp$ which is
the time necessary for the system to evolve from the initial state
$\psi_0$ to a state $\psi(T_\perp)$ such that ${\langle}
\psi_0,\psi(T_\perp){\rangle}=0$. (Here and in what follows by
$\lal\cdot,\cdot\ral$ we denote the inner product in the Hilbert
space $\fH$ assuming that it is linear in the first entry.)

The oldest among these bounds is the celebrated
\textit{Mandelstam--Tamm inequality} discovered in the 1945's paper
\cite{MaT}:
\begin{equation}
\label{MT} T_\perp\geq \frac{\pi}{2\, \Delta E},
\end{equation}
where $\Delta E$ is the energy dispersion for the initial
state $\psi_0$,
\begin{equation}
\label{DelE}
\Delta E=\bigl(\| H\psi_0\|^2 -{\langle} H\psi_0,\psi_0{\rangle}^2\bigr)^{1/2},
\quad \psi_0\in\mathrm{Dom}(H).
\end{equation}

The second celebrated lower bound for the orthogonalization time, the
\textit{Margolus--Levitin inequality} \cite{MaLe} has been
discovered half a century later, in 1998. This bound reads as
\begin{equation}
\label{ML} T_\perp\geq \frac{\pi}{2\, \delta E},
\end{equation}
where the quantity
\begin{equation}
\delta E= {\langle}
H\psi_0,\psi_0{\rangle}-\min\bigl(\mathop{\mathrm{spec}}(H)\bigr)
\end{equation}
represents the difference between the average energy for the state
$\psi_0$ and the lower edge of the spectrum of the Hamiltonian $H$
(which is assumed to be semibounded from below in this case).

The lower bounds \eqref{MT} and \eqref{ML} are not equivalent to
each other but both of them have been proven to be optimal (see,
e.g., \cite[p. 7]{DeCa} and \cite[p. 3923]{Frey}).

Of course, one notices that by their form the bounds \eqref{MT} and
\eqref{ML} resemble the uncertainty relation for energy and time.
These bounds, however, are related not to the standard deviation in
the measuring of the quantity $t$ but to the well-established time
needed for a state of the system to evolve into an orthogonal state.
Thus, in their essence the inequalities \eqref{MT} and \eqref{ML}
are very different from the uncertainty relation.

Next, there is a version of the Mandelstam--Tamm inequality
that works for intermediate time moments $t\in(0,T_\perp)$. This is
the lower estimate found for the first time in 1973 by Fleming \cite{Flem}:
\begin{equation}
\label{Flem} T_\theta\geq \frac{\theta}{\Delta E}\,,
\end{equation}
where $\Delta E$ is again given by \eqref{DelE} and $T_\theta$ denotes the
first time moment when the acute angle
\begin{equation}
\label{thetNaiv}
\angle\bigl(\psi_0,\psi(t)\bigr):=\arccos|{\langle}\psi_0,\psi(t){\rangle}|
\end{equation}
between the vectors $\psi_0$ and $\psi(t)$ reaches a certain value
$\theta\in(0,\pi/2]$.

It is worth to remark that, through the years, the Mandelstam-Tamm
bound \eqref{MT}/\eqref{Flem} has been rediscovered several times
(for related discussion and references, see, e.g.,
\cite[p.\,5]{DeCa}). The Mandelstam-Tamm bound has also been
extended to the evolution of mixed states \cite{Uhlmann}.
Furthermore, more detailed estimates for the evolution speed have
been established for particular classes of evolutionary problems
(see \cite{DeCa,Frey}). Probably, the latest among them is a speed
limit for evolution of thermal states derived in \cite{Lychko}.
Mandelstam-Tamm-type bounds for the orthogonalization time exist
even for some non-self-adjoint (so-called pseudo-Hermitian and, in
particular, $PT$-symmetric) Hamiltonians (see \cite{BBr,China} and
references therein).

In our recent work \cite{AM-PEPAN} we have generalized the
Mandelstam-Tamm-Fleming bound \eqref{Flem} to the Schr\"o\-din\-ger
evolution of a subspace. Like in the vast majority of publications
on quantum speed limits, in  \cite{AM-PEPAN} we restricted ourselves
to the exclusive consideration of bounded Hamiltonians. However typical
quantum-mechanical Hamiltonians are unbounded operators.  In the
present work we drop the requirement of boundedness of $H$ and
extend most of the results of \cite{AM-PEPAN} to the subspace
evolution governed by arbitrary self-adjoint Hamiltonians. Assuming
that $P_0$ is an orthogonal projection in $\fH$ such that the domain
of a (possibly unbounded) self-adjoint operator (Hamiltonian) $H$ is
invariant under $P_0$, that is, $P_0\,\Dom(H)\subset\Dom(H)$, we
study the subspace path $\fP_t=\Ran(P_t)$, $t\in\bbR$, formed in the
set of all the subspaces of $\fH$ by the ranges of the orthogonal
projections $P_t=\re^{-\ri Ht}P_0\re^{\ri Ht}$, $t\in\bbR$.

Our studies of the subspace evolution path $\fP_t$, $t\in\bbR$, are
essentially based on the concept of maximal angle\footnote{For
discussion of the concept of maximal angle and references see page
\pageref{Pthet}.} between two subspaces of a Hilbert space. Recall,
that the maximal angle between (arbitrary) subspaces $\fQ$ and $\fR$
of the Hilbert space $\fH$ is introduced as follows:
\begin{equation}
\label{thetmax}
\vartheta(\fQ,\fR)=\arcsin\|Q-R\|,
\end{equation}
where $Q$ and $R$ are the orthogonal projections in $\fH$ onto
$\fQ$ and $\fR$, respectively. The maximal angle \eqref{thetmax}
possesses all the properties of a distance, and thus it generates
a metric on the set of all subspaces of $\fH$. By using this metric we
establish, in particular, the following result (see Theorem
\ref{TTtheta} below).

Assume that $T_\theta$ is the time moment at which the maximal angle
between the initial subspace $\fP_0$ and a subspace in the path
$\fP_t$, $t\geq 0$, reaches a certain value $\theta\in(0,\pi/2]$.
Then necessarily
\begin{equation}
\label{thetInt} T_\theta\geq\frac{\theta}{\Delta
E_{\mathfrak{P}_0}}\,,
\end{equation}
where
\begin{equation}
\label{DeltaIn} \Delta
E_{\mathfrak{P}_0}=\sup\limits_{\scriptsize\begin{matrix}\psi\in\fP_0\cap\Dom(H)\\
\|\psi\|=1\end{matrix}} \bigl(\|H\psi\|^2-{\langle}
H\psi,\psi{\rangle}^2\bigr)^{1/2}.
\end{equation}
Clearly, the quantity $\Delta E_{\mathfrak{P}_0}$ is nothing but the
least upper bound for the energy dispersion on the states belonging
to the initial subspace $\fH_0.$

The Mandelstam-Tamm-Fleming bound \eqref{Flem} turns out to be a special
case of the estimate \eqref{thetInt} for a one-dimensional subspace
$\fP_0$ spanned by a particular state $\psi_0$.

The plan of the paper is as follows. In Section \ref{SecPropPath} we
collect some facts on the projection path $P_t=\re^{-\ri
Ht}P_0\re^{\ri Ht}$, $t\in\bbR$. In particular, we notice that this
path is strongly continuous on the whole Hilbert space $\fH$. Under
the assumption $P_0\,\Dom(H)\subset\Dom(H)$, the path $P_t$,
$t\in\bbR$, is, in addition, strongly differentiable in $t\in\bbR$
on the domain of $H$. In such a case the strong derivative
$\dot{P}_t$ is expressed through the commutator of $H$ and $P_t$
(see Theorem \ref{Lmain}). In Section \ref{SBounds} we work under
the additional hypothesis that the commutator of $H$ and $P_0$ is a
bounded operator on $\Dom(H)$ and, hence, its closure is a bounded
operator on the whole space $\fH$. The main result of the section is
Theorem \ref{Th2}. It presents the upper bound \eqref{TBound} for
the maximal angle between the subspaces $\fP_s$ and $\fP_t$,
$s,t\in\bbR$, in the subspace path $\fP_\tau=\Ran(P_\tau)$,
$\tau\in\bbR$, through the product of the times difference $|t-s|$
and the norm of the commutator of $H$ and $P_0$. This section also
contains the proof of Theorem \ref{TTtheta} that we already
mentioned.  The section is concluded with a new consideration of the
case where $H$ is bounded operator. In such a case the quantity
\eqref{DeltaIn} is bounded by half the distance between the upper
and lower edges of the spectrum of $H$. Combining this with
\eqref{thetInt} we obtain our last lower bound for $T_\theta$ (see
Theorem \ref{ThFin}) in this paper. Finally, in Section \ref{sConcl}
we present a summary of the work and point out some open problems.

Let us add a few words about notations used thro\-ug\-h\-o\-ut the
paper. By a subspace we always un\-der\-stand a closed linear subset
of a Hilbert space. The identity operator is denoted by $I$. For a
linear operator $L$, by $\Dom(L)$ we denote its domain and by
$\Ran(L)$, its range. If $Q$ is an orthogonal projection, the
notation $Q^\perp$ is always used for the complementary projection,
$Q^\perp=I-Q$. By $\fM\oplus\fN$ we understand the orthogonal sum of
two Hilbert spaces (or orthogonal subspaces or simply orthogonal
linear subsets) $\fM$ and $\fN$. By $\sE_T(\sigma)$ we always denote
the spectral projection of a self-adjoint operator $T$ associated
with a Borel set $\sigma\subset\bbR$. Notation $[A,B]$ is used for
the commutator of linear operators $A$ and $B$ on $\fH$. It is
assumed that
$\Dom\bigl([A,B]\bigr):=\bigl\{x\in\Dom(A)\cap\Dom(B)\,\big|\,\,
Ax\in\Dom(B),\, Bx\in\Dom(A)\bigr\}$ and $[A,B]x:=ABx-BAx$ \, for
any \, $x\in\Dom\bigl([A,B]\bigr)$.

\section{Projection path generated by the Schr\"odinger evolution
of a subspace}
\label{SecPropPath}

As it was already underlined, we are concerned not with a single
state but with a multi-dimensional (possibly, even
infinite-dimensional) subspace spanned by the states of the system
that are subject to the Schr\"odinger evolution. In general, we
allow the Hamiltonian $H$ of the system to be an {unbounded}
self-adjoint operator on the Hilbert space $\fH$ with domain
$\Dom(H)$. In this work we restrict ourselves to the consideration
of a nontrivial subspace $\mathfrak{P}_0\subset\mathfrak{H}$,
\mbox{$\mathfrak{P}_0\neq\{0\}$}, such that
\begin{equation}
\label{PDD} P_0\,\,\Dom(H)\subset\Dom(H),
\end{equation}
where $P_0$ stands for the orthogonal projection in $\fH$ onto
$\fP_0$. That is, we assume that the linear set $\Dom(H)$ is
invariant under $P_0$ in the sense that $P_0f\in\Dom(H)$ for any
$f\in \Dom(H)$. From \eqref{PDD} it follows that the set $\Dom(H)$
is also invariant under the complementary orthogonal projection
$P_0^\perp$,
\begin{equation}
\label{PDDp} P_0^\perp\,\,\Dom(H)\subset\Dom(H),
\end{equation}
Moreover, any of the hypotheses \eqref{PDD} and \eqref{PDDp} is
equivalent to any of the equalities
\begin{equation}
\label{main}
\Ran\left(P_0\big|_{\Dom(H)}\right)=\mathfrak{P}_0\cap\Dom(H)\quad
\text{and}\quad
\Ran\left(P^\perp_0\big|_{\Dom(H)}\right)=\mathfrak{P}^\perp_0\cap\Dom(H)
\end{equation}
as well as to the combination of them,
\begin{equation}
\label{basic}
\Dom(H)=\bigl(\fP_0\cap\Dom(H)\bigr)\,\,\oplus\,\,\bigl(\fP_0^\perp\cap\Dom(H)\bigr).
\end{equation}

Since $H$ is a self-adjoint operator, its domain $\Dom(H)$ is dense
in $\fH$. By \eqref{basic}, this implies that the sets
$\fP_0\cap\Dom(H)$ and $\fP_0^\perp\cap\Dom(H)$ are dense in the
subspaces $\fP_0$ and $\fP_0^\perp$, respectively. In particular,
\begin{equation}
\label{basdense} \overline{\fP_0\cap\Dom(H)}=\fP_0.
\end{equation}
From \eqref{main} it follows that the Hamiltonian $H$
admits the following $2\times 2$ block matrix representation with
respect to the decomposition $\fH=\fP_0\oplus\fP_0^\perp$:
\begin{align}
\label{Hmatrix}
H&=H_\mathrm{diag}+H_\mathrm{off},\quad
H_\mathrm{diag}=\left(\begin{array}{cc} H_{\fP_0} & 0
\\ 0 & H_{\fP_0^\perp}\end{array}\right), \quad H_\mathrm{off}=\left(\begin{array}{cc} 0  &
B \\ C & 0\end{array}\right),
\end{align}
where
\begin{align}
\label{HM1}
H_{\fP_0}&=P_0 H|_{\fP_0}, \quad \Dom(H_{\fP_0})=\fP_0\cap\Dom(H),\\
\label{HM2}
H_{\fP^\perp_0}&=P_0^\perp H|_{\fP^\perp_0},\quad
\Dom(H_{\fP^\perp_0})=\fP^\perp_0\cap\Dom(H),\\
\label{HM3}
B&=P_0 H|_{\fP^\perp_0},\quad \Dom(B)=\Dom(H_{\fP^\perp_0}),\\
\label{HM4}
C&=P_0^\perp H|_{\fP_0}, \quad \Dom(C)=\Dom(H_{\fP_0}).
\end{align}
Notice that, in general, $C\subset B^*$ and $B\subset C^*$.

Every vector $\psi_0\subset\mathfrak{P}_0\cap\Dom(H)$ is assumed to
evolve into a vector $\psi(t)\in\Dom(H)$, $t>0$,  according to
\eqref{Sch1p}, \eqref{Sch2p}. It is well known that for any
$\psi_0\in\Dom(H)$ the Cauchy problem \eqref{Sch1p}, \eqref{Sch2p}
has a unique solution $\psi: \mathbb{R}^+\to\fH$ in the class of
norm continuously differentiable vector-valued functions from
$\bbR^+$ to $\fH$ such that $\psi(t)\in\Dom(H)$ for all $t\geq 0$.
For convenience of the reader, we remark that the existence of the
solution to the problem \eqref{Sch1p}, \eqref{Sch2p} in the form
\begin{equation}
\label{U1} \psi(t)=U(t)\psi_0,\quad t>0,
\end{equation}
where
\begin{equation}
\label{UU1}
 U(t)=\mathrm{e}^{-\mathrm{i} Ht}, \quad t\in\mathbb{R},
\end{equation}
follows, e.g., from \cite[Theorem VIII.7]{ReedSimon:1}. Surely, the
exponential \eqref{UU1} is defined by the spectral theorem (see,
e.g., \cite[Theorem VIII.6]{ReedSimon:1}) as
\begin{equation}
\label{Uexp}
 \re^{-\ri Ht}:=\int_\bbR
e^{-\ri \lambda t} \sE(d\lambda),\quad t\in\bbR,
\end{equation}
where $\sE$ stands for the spectral measure on $\bbR$ associated
with the self-adjoint operator $H$. It is worth to notice that the
domain of $H$,
\begin{equation}
\label{DomH}
 \Dom(H)=\left\{f\in\fH\,\bigg|\,\,
\int_\bbR \lambda^2 \lal\sE(d\lambda)f,f\ral\right\},
\end{equation}
is an invariant of $U(t)$,
\begin{equation}
\label{Dinv}
\Ran\left(U(t)|_{\Dom(H)}\right)=\Dom(H),\quad \text{for
any $t\in\bbR$}.
\end{equation}

Given $t\in\bbR$, by $\mathfrak{P}_t$ we denote the range
$\Ran\left(U(t)P_0\right)$ of the product of the unitary operator
\eqref{UU1} and the orthogonal projection $P_0$ onto the subspace
$\fP_0$, that is,
\begin{equation}
\label{fPt}
 \fP_t:=\Ran\bigl(U(t)\big|_{\fP_0}\bigr), \quad t\in\bbR.
\end{equation}

By \eqref{U1}, the subspace $\fP_t$, $t>0$, is nothing but the
closure of the span of the vectors $\psi(t)\subset\Dom(H)$
representing the values, at the time moment $t$, of the
vector-valued functions $\psi:\,\,\bbR^+\to\fH$ that solve
\eqref{Sch1p}, \eqref{Sch2p} for various
$\psi_0\in\mathfrak{P}_0\cap\Dom(H)$.  So that we deal with the path
$ \mathfrak{P}_t,\,\,  t\geq 0,$ in the set of all subspaces of the
Hilbert space $\mathfrak{H}$. Or (and this is the same) with the
path
\begin{equation}
\label{Ppath} P_t, \quad t\geq 0,\qquad
\mathop{\mathrm{Ran}}\bigl(P_t\bigr)=\mathfrak{P}_t,
\end{equation}
of the orthogonal projections $P_t$ in $\mathfrak{H}$ onto the
respective subspaces $\mathfrak{P}_t$. Clearly, the orthogonal
projections $P_t$ onto the subspaces $\fP_t$ defined in \eqref{fPt}
are explicitly given by
\begin{equation}
\label{PT} P_t=U(t)P_0 U(t)^*=\mathrm{e}^{-\mathrm{i}
Ht}P_0\mathrm{e}^{\mathrm{i} Ht},\quad \text{for any $t\in\bbR$}.
\end{equation}
It is almost obvious that the strong continuity of the unitary group
$\re^{-\ri Ht}$, $t\in\bbR$, implies the strong continuity of the
path $P_t$, $t\in\bbR$, on the whole Hilbert space $\fH$.
Under the assumption \eqref{PDD} from \eqref{Dinv} it immediately
follows that the domain of $H$ is mapped by $P_t$ back into the
domain of $H$,
\begin{equation}
\label{PDom} \Ran\left(P_t\big|_{\Dom(H)}\right)\subset\Dom(H),
\quad t\in\bbR.
\end{equation}

For convenience of the reader we present a proof of the strong
differentiability of the projection family \eqref{PT} under the
hypothesis \eqref{PDD}.
\begin{theorem}
\label{Lmain}
Let $H$ be a (possibly unbounded) self-adjoint operator in the
Hilbert space $\fH$. Assume that $P_0$, $P_0\neq 0$, is an
orthogonal projection in $\fH$ and that the domain of $H$ is
invariant under $P_0$, i.e.,
$\Ran\bigl(P_0|_{\Dom(H)}\bigr)\subset\Dom(H)$. Then the projection
path $P_t=\re^{-\ri Ht}P_0\re^{\ri Ht}$, $t\in\bbR$, is strongly
differentiable on $\Dom(H)$ for any $t\in\bbR$, that is, the
following limit exists
\begin{align}
\label{Pder} \dot{P}_tf:=&\lim\limits_{\tau\to 0}
\left(\frac{P_{t+\tau}-P_t}{\tau}f\right),\quad{t\in\bbR},\\
\nonumber
&\text{for any } f\in\Dom(H).
\end{align}
Moreover, the inclusion \eqref{PDom} holds
and the following equality takes place:
\begin{align}
\label{UPSh} \ri \dot{P}_t f = &HP_tf-P_tHf, \quad{t\in\bbR},\\
\nonumber
&\text{for any } f\in\Dom(H).
\end{align}
\end{theorem}
\begin{proof}
As we already noticed, under the hypothesis
$\Ran\bigl(P_0|_{\Dom(H)}\bigr)\subset\Dom(H)$ the inclusion
\eqref{PDom} follows immediately from \eqref{Dinv}. Assume that
$f\in\Dom(H)$ and $t\in\bbR$. By \eqref{PDom} we have
$P_tf\in\Dom(H)$. Now take $0\neq\tau\in\bbR$ and write
\begin{align}
\nonumber
\frac{P_{t+\tau}-P_t}{\tau}f&
=\frac{\re^{-\ri H(t+\tau)}P_0\re^{\ri H(t+\tau)}-\re^{-\ri Ht}P_0\re^{\ri Ht}}{\tau}f\\
\nonumber
&=\frac{\re^{-\ri H(t+\tau)}P_0\re^{\ri H(t+\tau)}-\re^{-\ri H(t+\tau)}P_0\re^{\ri Ht}
+\re^{-\ri H(t+\tau)}P_0\re^{\ri Ht}-\re^{-\ri Ht}P_0\re^{\ri Ht}}{\tau}f\\
\label{gsum12}
&=g_1(t,\tau,f)+g_2(t,\tau,f),
\end{align}
where the vectors $g_1$ and $g_2$ are given by
\begin{align*}
g_1(t,\tau,f)&=
\re^{-\ri H(t+\tau)}P_0 \re^{\ri Ht}\, \frac{\re^{\ri H\tau}-I}{\tau}f,\\
g_2(t,\tau,f)&=\re^{-\ri Ht}\frac{\re^{-\ri H\tau}-I}{\tau}P_0\re^{\ri Ht}f.
\end{align*}
For $g_1(t,\tau,f)$ we have
\begin{align}
\label{g1sum}
g_1(t,\tau,f)&=g_{11}(t,\tau,f)+g_{12}(t,\tau,f),
\end{align}
where
\begin{align*}
g_{11}(t,\tau,f)&=\re^{-\ri H(t+\tau)}P_0 \re^{\ri Ht}\,
\left(\frac{\re^{\ri H\tau}-I}{\tau}f-\ri Hf\right),\\
g_{12}(t,\tau,f)&=\re^{-\ri H(t+\tau)}P_0 \re^{\ri Ht}\, \left(\ri Hf\right).
\end{align*}
The unitarity of $\re^{-\ri H(t+\tau)}$ and $\re^{\ri Ht}$ jointly with the fact
that the orthogonal projection $P_0$ has unit norm, $\|P_0\|=1$, yields
\begin{align}
\label{g11}
\|g_{11}(t,\tau,f)\|\leq \left\|\re^{-\ri H(t+\tau)}\right\|
\left\|P_0\right\|\left\|\re^{\ri Ht}\right\|\left\|\frac{\re^{\ri
H\tau}-I}{\tau}f-\ri Hf\right\|\leq \left\|\frac{\re^{\ri
H\tau}-I}{\tau}f-\ri Hf\right\|.
\end{align}
Taking into account the strong differentiability of the group $\re^{\ri Ht}$ on
$\Dom(H)$, from  \eqref{g11} one concludes that
\begin{equation}
\label{g11l}
g_{11}(t,\tau,f)\to 0\quad\text{as}\quad \tau\to 0 \quad
(\text{for any $t\in\bbR$ and any $f\in\Dom(H)$}).
\end{equation}
At the same time, in view of the strong continuity of the group
$\re^{\ri Ht}$ on $\fH$,
\begin{equation}
\label{g12}
g_{12}(t,\tau,f)\mathop{\longrightarrow}\limits_{\tau\to
0} \re^{-\ri Ht}P_0 \re^{\ri Ht}\, \left(\ri Hf\right)= \ri P_t
Hf\quad (\text{for any $t\in\bbR$ and any $f\in\Dom(H)$}).
\end{equation}

As for the term $g_2(t,\tau,f)$, it is easy to observe that
\begin{align*}
g_2(t,\tau,f)& \mathop{\longrightarrow}\limits_{\tau\to 0} \re^{-\ri
Ht}(-\ri H)P_0\re^{\ri Ht}f \quad (\text{for any $t\in\bbR$ and any $f\in\Dom(H)$}).
\end{align*}
This follows again from the strong differentiability of the group
$\re^{\ri Ht}$ on $\Dom(H)$, taking into account that $\re^{\ri
Ht}f\in\Dom(H)$ whenever $f\in\Dom(H)$ and then $P_0\re^{\ri
Ht}f\in\Dom(H)$ by the hypothesis. Now it only remains to recollect
that $ \re^{-\ri Ht}(-\ri H)P_0\re^{\ri Ht}f=-\ri H \re^{-\ri
Ht}P_0\re^{\ri Ht}f= -\ri H P_tf, $ which means that
\begin{align}
\label{g2l} g_2(t,\tau,f)& \mathop{\longrightarrow}\limits_{\tau\to
0}  -\ri H P_tf \quad (\text{for any $t\in\bbR$ and any
$f\in\Dom(H)$}).
\end{align}
Combining this result with \eqref{gsum12}, \eqref{g1sum},
\eqref{g11l}, and \eqref{g12} completes the proof.
\end{proof}
\begin{remark}
\label{RHPstrong} Equality \eqref{UPSh} implies that, under the
hypothesis of Theorem \ref{Lmain}, the projection path \eqref{PT} is
a strong solution (on $\Dom(H)$) to the Cauchy problem
\begin{align}
\label{Sch1P}
\mathrm{i} \frac{d}{dt} P_t & =[H,P_t],\\
\label{Sch2P} P_t\big|_{t=0}&=P_0,
\end{align}
where $\frac{d}{dt}P_t=\dot{P}_t$ stands for the strong derivative
\eqref{Pder} and
\begin{equation}
\label{HPPH}
[H,P_t]:=HP_t-P_tH, \quad \Dom([H,P_t])=\Dom(H),
\end{equation}
denotes the commutator of $H$ and $P_t$.
\end{remark}
The following observation may be helpful in the study of variation
of the subspaces $\fP_t$ defined in \eqref{fPt}.
\begin{remark}
\label{Roff} The inclusion \eqref{PDom} implies that also
$\Ran\left(P^\perp_t\big|_{\Dom(H)}\right)\subset\Dom(H)$. This
means that, at any moment $t\in\bbR$, the commutator \eqref{HPPH}
may be written in the form
\begin{equation}
\label{HPPH1}
[H,P_t]=P^\perp_t HP_t-P_tHP^\perp_t.
\end{equation}
From \eqref{HPPH1} it follows that, with respect to the orthogonal
decomposition $\fH=\fP_t\oplus\fP^\perp_t$,  the operator $[H,P_t]$
admits representation in the form of the following $2\times2$  block
off-diagonal skew-symmetric operator matrix:
\begin{equation}
\label{HPPH2}
[H,P_t]=\begin{pmatrix}
0 & -P_t H\bigr|_{\fP^\perp_t}\\
P^\perp_t H\bigr|_{\fP_t} & 0
\end{pmatrix}, \quad t\in\bbR,
\end{equation}
which is considered on the domain
$\bigl(\Dom(H)\cap\fP_t\bigr)\oplus
\bigl(\Dom(H)\cap\fP^\perp_t\bigr)=\Dom(H)$.
\end{remark}

\begin{remark}
\label{uneq} One more observation is that for any $t\in\bbR$
the commutator \eqref{HPPH} of $H$ and $P_t$ remains unitary
equivalent to its value $[H,P_0]$ at $t=0$, namely,
\begin{equation}
\label{euneq} [H,P_t]=\mathrm{e}^{-\mathrm{i}
Ht}[H,P_0]\mathrm{e}^{\mathrm{i} Ht}, \quad t\in\bbR.
\end{equation}
\end{remark}

\section{Bounds for the speed of the subspace evolution}
\label{SBounds}
In this section we again assume that $H$ is a self-adjoint operator
on the Hilbert space $\fH$ and that $P_0$ is an orthogonal
projection on $\fH$ with $\Ran(P_0)=\fP_0\neq\{0\}$ satisfying the
condition \eqref{PDD}.

It is well known that the set of all orthogonal projections in the
Hilbert space $\mathfrak{H}$ (and hence the set of all subspaces of
$\mathfrak{H}$) is a metric space with the distance given by the
operator norm,
\begin{equation}
\label{dnaiv}
{\rho}(Q,R):=\|Q-R\|, \qquad
{\rho}(\mathfrak{Q},\mathfrak{R}):={\rho}(Q,R),
\end{equation}
where $Q$, $R$ are arbitrary orthogonal projections and
$\fQ$, $\fR$, their respective ranges.
It is worth mentioning that
$$
\|Q-R\|\leq 1
$$
for any two orthogonal projections $Q$ and $R$ in $\fH$ (see, e.g.,
\cite[Section 34]{AkhiG}) and, thus, we always have $\rho(\fQ,\fR)\leq
1$ for any subspaces $\fQ$ and $\fR$ in $\fH$.

We start with proving our first quantum speed limit for the subspace
variation based on the metric \eqref{dnaiv}. This ``na\"ive'' speed
limit follows directly from the equation \eqref{UPSh} under an
additional assumption that the commutator of the operators $H$ and
$P_0$ is a bounded operator on its domain
$\Dom\bigl([H,P_0]\bigr)=\Dom(H)$.

\begin{lemma}
\label{Lnaive} Assume the hypothesis of Theorem \ref{Lmain}. Assume,
in addition, that the commutator $[H,P_0]$ considered on
$\Dom\bigl([H,P_0]\bigr)=\Dom(H)$ is a bounded operator, that is,
\begin{equation}
\label{VHPdef}
V_{{H,P_0}}:=\mathop{\sup}\limits_{\scriptsize\begin{matrix}f\in\Dom(H)\\
\|f\|=1\end{matrix}}\|HP_0f-P_0Hf\|<\infty,
\end{equation}
and let $P_t=\re^{-\ri Ht}P_0\re^{\ri Ht}$, $t\in\bbR$. Then the
closure $\overline{[H,P_t]}$ of the commutator $[H,P_t]$,
$t\in\bbR$, is a bounded operator on the whole Hilbert space $\fH$
and
\begin{equation}
\label{HPV} \bigl\|\overline{[H,P_t]}\bigr\|=V_{{H,P_0}}
\quad\text{for any \,}t\in\bbR.
\end{equation}
Furthermore, the following inequality holds
\begin{equation}
\label{MBound}
\|P_{t}-P_s\|\leq
V_{{H,P_0}}\, |t-s|, \quad \text{for any \,}t,s\in\bbR.
\end{equation}
\end{lemma}
\begin{proof}
Since $H$ is a self-adjoint operator, its domain is dense in $\fH$.
Then the boundedness \eqref{VHPdef} of the commutator $[H,P_0]$
on $\Dom(H)$ implies that its closure $\overline{[H,P_0]}$ is a
bounded operator defined on the whole Hilbert space $\fH$.
Continuity of the norm implies
$\bigl\|\overline{[H,P_0]}\bigr\|=\|[H,P_0]\|=V_{{H,P_0}}$. By
Remark \ref{uneq}, for any $t\in\bbR$ the commutator $[H,P_t]$ is
unitary equivalent to the commutator $[H,P_0]$ and the same concerns
their closures. Hence, $\overline{[H,P_t]}$ is a bounded operator on
$\fH$ with the norm coinciding with that of $\overline{[H,P_0]}$,
and this proves \eqref{HPV}.

Now assume that $s,t\in\bbR$ and, for definiteness, let  $s<t$.
Since $\overline{\Dom(H)}=\fH$, for the norm of $P_t-P_s$ we have
\begin{equation}
\label{PPnorm}
\|P_t-P_s\|=\mathop{\sup}\limits_{\scriptsize\begin{matrix}f\in\fH\\
\|f\|=1\end{matrix}}\|P_t f-P_s f\|=
\mathop{\sup}\limits_{\scriptsize\begin{matrix}f\in\Dom(H)\\
\|f\|=1\end{matrix}}\|P_t f-P_s f\|.
\end{equation}
Then by Theorem \ref{Lmain} (in particular, by the strong
differentiability \eqref{Pder} of the path $P_t$) one concludes that
\begin{equation}
\label{PPnorm1}
\|P_t-P_s\|=
\mathop{\sup}\limits_{\scriptsize\begin{matrix}f\in\Dom(H)\\
\|f\|=1\end{matrix}}\left\|\int_s^t \dot{P}_\tau f d\tau\right\|
\leq \int_s^t \mathop{\sup}\limits_{\scriptsize\begin{matrix}f\in\Dom(H)\\
\|f\|=1\end{matrix}}\left\|\dot{P}_\tau f \right\|d\tau.
\end{equation}
Now by taking into account \eqref{UPSh} and \eqref{HPV} one arrives at
\begin{equation}
\label{PPnorm2}
\|P_t-P_s\|\leq
\int_s^t \mathop{\sup}\limits_{\scriptsize\begin{matrix}f\in\Dom(H)\\
\|f\|=1\end{matrix}}\|HP_\tau f-P_\tau Hf\| d\tau
= \int_s^t \bigl\|\overline{[H,P_\tau]}\bigr\| d\tau
=V_{{H,P_0}}\, (t-s),
\end{equation}
which completes the proof.
\end{proof}

Soon we will see that there is a stronger estimate of $\|P_t-P_s\|$
through $|t-s|$ than the one presented in \eqref{MBound}. The
stronger estimate (see inequality \eqref{TBound} in Theorem
\ref{Th2} below) is associated with another natural but much less
known metric on the set of all subspaces of $\fH$. This metric is
associated with the quantity \label{Pthet}
\begin{equation}
\label{rhoM}
\vartheta(\fQ,\fR):=\arcsin(\|Q-R\|)
\end{equation}
called the \textit{maximal angle} between the subspaces $\fQ$ and
$\fR$. The fact that the maximal angle \eqref{rhoM} is a metric has
been proven in 1993 by L.Brown \cite{Brown}. An alternative proof of
this fact may be found in \cite{AM-CAOT}.
Let us also refer to the discussion of the metric \eqref{rhoM} in
\cite{MS2015}.

Notice that, since the maximal angle is a metric, the triangle
inequality
\begin{align}
\label{trian} \vartheta(\fQ,\fR)&\leq \vartheta(\fQ,\mathfrak{S}) +
\vartheta(\mathfrak{S},\fR)
\end{align}
holds for any subspaces $\fQ,\fR,\mathfrak{S}\subset\fH$.

\begin{remark}
\label{rhoth} Inequality $x<\arcsin x$, $x\in(0,1]$, implies that
always $\rho(\fQ,\fR)<\vartheta(\fQ,\fR)$  if $\fQ\neq\fR$. Thus,
the metric $\vartheta$  is stronger than the metric $\rho$ in the
sense that the bound $\vartheta(\fQ,\fR)< c$ for some $c>0$
automatically requires that also $\rho(\fQ,\fR)<c$. The converse is
not true in general.
\end{remark}

\begin{remark}
\label{newr} One verifies by inspection that the non-negative
operator $(Q-R)^2$ is block diagonal with respect to the
decomposition  $\fH=\fR\oplus\fR^\perp$, more precisely,
\begin{equation}
(Q-R)^2=RQ^\perp R+R^\perp Q R^\perp,
\end{equation}
which means that
\begin{equation}
\label{impr}
\|Q-R\|=\max\bigl\{\|RQ^\perp R\|^{1/2},\|R^\perp Q R^\perp\|^{1/2}\bigr\}=
\max\bigl\{\|Q^\perp R\|, \|R^\perp Q\|\bigr\}.
\end{equation}
\end{remark}

\begin{remark}
The concept of maximal angle between subspaces can be traced back at
least to Krein, Krasnoselsky, and Milman \cite{KKM1948}. Under the
assumption that $(\mathfrak{Q},\mathfrak{R})$ is an ordered pair of
subspaces and $\mathfrak{Q}\neq\{0\}$, the notion of the (relative)
maximal angle between $\mathfrak{Q}$ and $\mathfrak{R}$ is applied
in  \cite{KKM1948} to the number
$\varphi(\mathfrak{Q},\mathfrak{R})\in[0,\pi/2]$ defined by
\begin{equation}
\label{t12}
\sin\varphi(\mathfrak{Q},\mathfrak{R})=\sup\limits_{x\in\mathfrak{Q},\,\|x\|=1}\mathop{\rm
dist}(x,\mathfrak{R}).
\end{equation}
Obviously, equality \eqref{t12} is equivalent to
\begin{equation}
\label{impr1}
\sin\varphi(\mathfrak{Q},\mathfrak{R})=\|Q^\perp R\|.
\end{equation}
If both $\mathfrak{Q}\neq\{0\}$ and $\mathfrak{R}\neq\{0\}$ then
\eqref{impr} implies
\begin{equation}
\label{tet}
\vartheta(\mathfrak{Q},\mathfrak{R})=\arcsin\bigl(\max\bigl\{\|Q^\perp
R\|,\|R^\perp Q\|\}\bigr)=\max\bigl\{\varphi(\mathfrak{Q},\mathfrak{R}),
\varphi(\mathfrak{R},\mathfrak{Q})\bigr\}.
\end{equation}
In contrast to $\varphi(\mathfrak{Q},\mathfrak{R})$, the maximal angle
$\vartheta(\mathfrak{Q},\mathfrak{R})$ is always symmetric with
respect to the interchange of the arguments $\mathfrak{Q}$ and
$\mathfrak{R}$. Moreover,
\begin{equation}
\label{P-Q1}
\varphi(\mathfrak{R},\mathfrak{Q})=\varphi(\mathfrak{Q},\mathfrak{R})=
\vartheta(\mathfrak{Q},\mathfrak{R})\quad \text{whenever
}\|Q-R\|<1.
\end{equation}
For $\|Q-R\|<1$, this is simply a consequence of the equality
 $\|Q^\perp R\|=\|R^\perp Q\|$, which is
easily deduced, e.g., from \cite[Corollary 3.4 (i) and Remark
3.6]{KMM2003}.
\end{remark}

\begin{remark}
\label{PhSense}
The maximal angle between subspaces admits a quantum-mechanical
interpretation. To this end, one may apply the concept of a
subspace-state of a quantum system. Namely, given a subspace
$\mathfrak{Q}\subset\mathfrak{H}$, one says that the system is in
the $\mathfrak{Q}$-state if it is in a pure state described by a
(non-specified) normalized vector $x\in\mathfrak{Q}$. By \eqref{t12}
and \eqref{tet} the quantity
$\cos^2\vartheta(\mathfrak{Q},\mathfrak{R})$ is then treated as a
minimum probability for a quantum system which is in a
$\mathfrak{Q}$-state to be found also in an $\mathfrak{R}$-state.
\end{remark}

By using the maximal-angle metric \eqref{rhoM} we obtain the following result.

\begin{theorem}
\label{Th2} Assume the hypothesis of Lemma \ref{Lnaive} and let
$\mathfrak{P}_\tau= \mathop{\mathrm{Ran}}\bigl(P_\tau\bigr)$,
$\tau\in\bbR$.
Then the following inequality holds:
\begin{equation}
\label{TBound}
\vartheta\bigl(\mathfrak{P}_s,\mathfrak{P}_t\bigr)\leq
V_{{H,P_0}}\,|t-s| \quad \text{for any \,}s,t\in\bbR.
\end{equation}
\end{theorem}
\begin{proof}
Assume, for definiteness, that $s\neq t$ and set
\begin{equation}
\label{tauj}
\tau_j=s+j\,\,\frac{t-s}{n},\quad j=0,1,\ldots,n,
\end{equation}
where $n$ is a natural number, $n\in\bbN$. Notice that $\tau_0=s$
and $\tau_n=t$. Under the assumption that $n\geq 2$, by applying the
triangle inequality \eqref{trian} to the subspaces $\fP_s$, $\fP_t$,
and intermediate subspaces $\fP_{\tau_j}$, $j=1,2,\ldots,n-1$, one
arrives at the following  bound for the maximal angle
$\vartheta(\fP_s,\fP_t)$ between the subspaces $\fP_s$ and $\fP_t$:
\begin{equation}
\label{tQbound}
\vartheta(\fP_0,\fP_t) \leq
\sum_{j=1}^n \arcsin \bigl\|P_{\tau_j}-P_{\tau_{j-1}}\bigr\|.
\end{equation}
By \eqref{tauj} we have
\begin{equation}
\label{taun}
|\tau_j-\tau_{j-1}|=\frac{|t-s|}{n},\quad j=1,2,\ldots,n.
\end{equation}
Now take $n$ such that $V_{{H,P_0}}\frac{|t-s|}{n}\leq 1$. Then
combining  \eqref{tQbound} and \eqref{taun} with the estimate
\eqref{MBound} in Lemma \ref{Lnaive} yields
\begin{align}
\vartheta(\fP_0,\fP_t) \leq&
\sum_{j=1}^n \arcsin \bigl(V_{{H,P_0}}\, |\tau_j-\tau_{j-1}|\bigr)
\label{tQbound2}
= n\,\arcsin \frac{V_{{H,P_0}}|t-s|}{n}.
\end{align}
Passing in \eqref{tQbound2} to the limit $n\to\infty$ one arrives at
\eqref{TBound}, completing the proof.
\end{proof}
\begin{remark}
\label{Rem-on-naive} Given $s\neq t$ and $V_{P_0,H}>0$, the bound
\eqref{TBound} is more tight for $\|P_s-P_t\|$ than the bound
\eqref{MBound} (cf. Remark \ref{rhoth}).
\end{remark}
\begin{remark}
\label{RVHPp} Under the hypothesis of Lemma \ref{Lnaive}, by Remark
\ref{Roff} from \eqref{VHPdef} it follows that
\begin{equation}
\label{VPt} V_{{H,P_t}}=\|{P_t H P_t^\perp}\|=\|{P_t^\perp HP_t}\|\quad
(=\|\overline{P_t H P_t^\perp}\|=\|\overline{P_t^\perp
HP_t}\|)\quad\text{for any } t\in\bbR
\end{equation}
and, in particular,
\begin{equation}
\label{VHP} V_{{H,P_0}}=\|{P_0HP_0^\perp}\|=\|{P_0^\perp H P_0}\|.
\end{equation}
Therefore, the bound \eqref{TBound} may be interpreted in the sense
that only the off-diagonal entries  $P_0 H\bigr|_{\fP_0^\perp}$ and
$ P_0^\perp H\bigr|_{\fP_0}$ in the block matrix
representation \eqref{Hmatrix} of the Hamiltonian $H$ contribute
into the variation of the subspace $\mathfrak{P}_0$. If $H$ is block
diagonal with respect to the decomposition
$\mathfrak{H}=\mathfrak{P}_0\oplus\mathfrak{P}_0^\perp$ and, hence,
the subspace $\mathfrak{P}_0$ is reducing for $H$, it does not change with
time at all. In particular, none of the spectral subspaces of $H$
can be a subject of the time evolution.
\end{remark}

\begin{corollary}
\label{Cor1} Under the hypothesis of Lemma \ref{Lnaive}, suppose that
$T_\theta$ is the time when the maximal angle between the
initial subspace $\mathfrak{P}_0$ and a subspace in the path
$\mathfrak{P}_t$, $t\geq 0$, reaches the value of $\theta$,
$0<\theta\leq \frac{\pi}{2}$, that is,
\begin{equation}
\label{TtP}
\vartheta\bigl(\mathfrak{P}_0,\mathfrak{P}_{T_\theta})=\theta.
\end{equation}
Then\footnote{\label{pref1}For the case where $V_{H,P_0}=0$, which
implies that $\fP_0$ is a reducing subspace of $H$, one adopts the
convention that $T_\theta=\infty$.}
\begin{equation}
\label{TtPb} T_\theta\geq \frac{\theta}{V_{{H,P_0}}}\,.
\end{equation}
\end{corollary}

\begin{example}
\label{exmpl} Let the Hamiltonian $H$ describe a two-level quantum
system with bound states $e_1$ and $e_2$, that is,
$\|e_1\|=\|e_2\|=1$, ${\langle} e_1,e_2{\rangle}=0$, the Hilbert
space $\fH$ is the span of the vectors $e_1$ and $e_2$, and
$$
H=E_1 {\langle} \cdot,e_1{\rangle} e_1+ E_2 {\langle}
\cdot,e_2{\rangle} e_2,
$$
where the binding energies $E_1$ and $E_2$ are supposed to be
different, $E_1\neq E_2$. Assume that $P_0$ is the orthogonal
projection on the one-dimensional subspace $\fP_0$ spanned by the
vector $e=\frac{1}{\sqrt{2}}(e_1+e_2)$. One immediately observes
that
\begin{equation}
\label{HPex}
[H,P_0]=\frac{E_2-E_1}{2}\bigl(\langle \cdot,e_1\rangle e_2-\langle \cdot,e_2\rangle e_1\bigr)
\end{equation}
and then for the norm of the commutator of $H$ and $P_0$ we get
\begin{equation}
\label{HPexnorm}
V_{H,P_0}=\bigl\|[H,P_0]\bigr\|=\frac{|E_2-E_1|}{2}.
\end{equation}
One also verifies by inspection that this norm may be written as
\begin{equation}
\label{HPnormEE}
V_{H,P_0}=\bigl(\|He\|^2-\langle He,e\rangle^2\bigr)^{1/2}.
\end{equation}
Furthermore, an elementary computation shows that  for any
$\tau\in\bbR$ the orthogonal projection $P_\tau=\re^{-\ri
H\tau}P_0\re^{\ri H\tau}$ is given by
\begin{equation}
\label{Ptau}
P_\tau=\frac{1}{2}\left(
\langle\cdot,e_1\rangle e_1+
\re^{-\ri(E_2-E_1)\tau}\langle \cdot,e_1\rangle e_2
+\re^{\ri(E_2-E_1)\tau}\langle \cdot,e_2\rangle e_1
+\langle\cdot,e_2\rangle e_2
\right).
\end{equation}
Then for any $s,t\in\bbR$ we have
\begin{equation}
\label{Pst}
P_t-P_s=\frac{1}{2}\left(
\re^{-\ri(E_2-E_1)t}-\re^{-\ri(E_2-E_1)s}\right)\langle \cdot,e_1\rangle e_2+
\frac{1}{2}\left(\re^{\ri(E_2-E_1)t}-\re^{\ri(E_2-E_1)s}\right)\langle \cdot,e_2\rangle e_1.
\end{equation}
The eigenvalues of the rank-two operator
\eqref{Pst} are
\begin{equation}
\label{lampm}
\lambda_\pm=\pm\sin\left(\frac{|E_2-E_1|}{2}|t-s|\right),
\end{equation}
which means in particular that
\begin{equation}
\label{ThetEx}
\vartheta(\fP_s,\fP_t)=\arcsin\|P_t-P_s\|=\frac{|E_2-E_1|}{2}|t-s|\,\text{\, whenever }\frac{|E_2-E_1|}{2}|t-s|\leq\frac{\pi}{2},
\end{equation}
where, as usually, $\fP_\tau=\Ran(P_\tau)$, $\tau\in\bbR$.
\end{example}

\begin{remark}
Example \ref{exmpl} is used in many papers on quantum speed limits
(see, e.g., \cite{DeCa,Frey,BBr,China}). In particular, this example
proves the sharpness of both the Mandelstam-Tamm and Margolus-Levitin
inequalities (see, e.g., \cite[Section 2.4]{DeCa}). Example \ref{exmpl}
also works for the bounds \eqref{TBound} and \eqref{TtPb}.  In this
example, due to \eqref{HPexnorm} and \eqref{ThetEx} both of
these bounds transform into equalities, which proves that both the
bounds \eqref{TBound} and \eqref{TtPb} are optimal.
\end{remark}

\begin{theorem}
\label{TTtheta} Assume the hypothesis of Theorem \ref{Lmain} and let
$\fP_\tau=\Ran\bigl(\re^{-\ri H\tau}P_0\re^{\ri H\tau}\bigr)$,
$\tau\in\bbR$.
Assume, in addition, that
\begin{equation}
\label{Delta} \Delta
E_{\mathfrak{P}_0}:=\sup\limits_{\scriptsize\begin{matrix}f\in\fP_0\cap\Dom(H)\\
\|f\|=1\end{matrix}} \bigl(\|Hf\|^2-{\langle}
Hf,f{\rangle}^2\bigr)^{1/2}<\infty.
\end{equation}
Then
\begin{equation}
\label{TBoundE}
\vartheta\bigl(\mathfrak{P}_s,\mathfrak{P}_t\bigr)\leq
\Delta
E_{\mathfrak{P}_0}\,|t-s| \quad \text{for any \,}t,s\in\bbR,
\end{equation}
and\footnote{\label{pref2}For the case where $\Delta
E_{\mathfrak{P}_0}=0$, which implies that $\fP_0$ is the
eigenspace associated with an eigenvalue of $H$, one adopts the
convention that $T_\theta=\infty$ (cf. footnote on page
\pageref{pref1}).}
\begin{equation}
\label{TtPb1} T_\theta\geq \frac{\theta}{\Delta
E_{\mathfrak{P}_0}}\,,
\end{equation}
where $\theta$, $T_\theta$ are the same as in Corollary \ref{Cor1}.
\end{theorem}
\begin{proof}
We start with the elementary observation that
\begin{equation}
\label{TPif} \|P_0^\perp Hf\|=\bigl(\|Hf\|^2-\|P_0Hf\|^2\bigr)^{1/2}\quad
\text{for any }f\in\Dom(H).
\end{equation}
If, in addition, $f\in\fP_0$ and $\|f\|=1$ then automatically $\|P_0y\|\geq|\langle
y,f\rangle|$  for any $y\in\fH$. In particular,
\begin{equation}
\label{Py}
\|P_0 Hf\|\geq\langle Hf,f\rangle\quad\text{for any }f\in\fP_0\cap\Dom(H),\, \|f\|=1.
\end{equation}
Hence, from \eqref{TPif} it follows that
\begin{equation}
\label{TPif1} \|P_0^\perp Hf\|\leq
\bigl(\|Hf\|^2-\langle Hf,f\rangle^2\bigr)^{1/2}\quad
\text{for any }f\in\fP_0\cap\Dom(H), \,\|f\|=1.
\end{equation}
Meanwhile, by  Remark \ref{RVHPp} (see second equality in \eqref{VHP}) we have
\begin{equation}
\label{VHPbf1}
V_{H,P_0}=\sup\limits_{\scriptsize\begin{matrix}f\in\Dom(H)\\
\|f\|=1\end{matrix}}\|P_0^\perp H P_0 f\|=
\sup\limits_{\scriptsize\begin{matrix}f\in\fP_0\cap\Dom(H)\\
\|f\|=1\end{matrix}}\|P_0^\perp H f\|,
\end{equation}
by taking into account that, by the hypothesis, the linear set $\Dom(H)$
is invariant under $P_0$ and then $P_0\Dom(H)=\fP_0\cap\Dom(H)$ (see
equalities \eqref{PDD}--\eqref{main}).
By combining \eqref{VHPbf1} with \eqref{Delta} and \eqref{TPif1}, one obtains
\begin{equation}
\label{VHPbf2}
V_{H,P_0}\leq  \Delta E_{\mathfrak{P}_0}.
\end{equation}
Then \eqref{TBoundE} follows from the estimate \eqref{TBound} in
Theorem \ref{Th2}, while \eqref{TtPb1} is implied by the bound
\eqref{TtPb} in Corollary \ref{Cor1}. The proof is complete.
\end{proof}

\begin{remark}
We underline that the above proof is based on the inequality
\eqref{VHPbf2}. Due to this inequality, the hypothesis of Theorem
\ref{TTtheta} implies the one of Theorem \ref{Th2} and the bound
\eqref{TtPb} implies the bound \eqref{TtPb1}. In this sense, Theorem
\ref{TTtheta} is a weaker version of Theorem \ref{Th2} but the bound
\eqref{TtPb1} much closer resembles a Mandelstam-Tamm-Fleming bound
\eqref{Flem}. (Also notice that in the case of a one-dimensional
subspace $\fP_0$ inequality \eqref{VHPbf2} turns into equality and
both the bounds \eqref{TtPb} and \eqref{TtPb1} reduce to the
Mandelstam-Tamm-Fleming bound \eqref{Flem} with the energy
dispersion \eqref{DelE}  involving a state vector $\psi_0$ spanning
$\fP_0$.)
\end{remark}

\begin{remark}
\label{remref}
Combining equality \eqref{ThetEx} with equalities \eqref{HPexnorm}
and \eqref{HPnormEE} shows that in Example \ref{exmpl} both bounds \eqref{TBoundE} and
\eqref{TtPb1} turn into precise equalities. Thus, Example
\ref{exmpl} proves that these bounds are sharp.
\end{remark}

The following proposition is verified by direct inspection.

\begin{proposition}
\label{LTT}
Assume that $T$ is a linear operator on a Hilbert space $\fH$ with
domain $\Dom(T)$. Than the following identity holds
\begin{align}
\label{Id1}
&\| Tx\|^2-|\lal Tx,x\ral|^2=\|(T-cI)x\|^2-|\lal (T-cI)x,x\ral|^2\\
\nonumber
&\hspace*{0.3cm}\text{ for any\, $c\in\mathbb{C}$ and any\, $x\in\Dom(T)$, $\|x\|=1$}.
\end{align}
\end{proposition}

The statement below is rather well known. We present and prove it
only for convenience of the reader.

\begin{lemma}
\label{cMm}
Let $T$ be a bounded self-adjoint operator on the Hilbert space $\fH$. Then the
following inequalities hold:
\begin{equation}
\label{AMm}
0\leq \|Tx\|^2-\lal Tx,x\ral^2\leq \frac{1}{4}(M-m)^2, \quad \text{for any\, $x\in\fH$, $\|x\|=1$,}
\end{equation}
where $m=\min\bigl(\spec(T)\bigr)$ and $M=\max\bigl(\spec(T)\bigr)$ are the upper and
lower bounds of the spectrum of $T$, respectively.
\end{lemma}
\begin{proof}
The left inequality in \eqref{AMm} is obvious since $|\lal
Tx,x\ral|\leq \|Tx\|\|x\|=\|Tx\|$ because of $\|x\|=1$. Now
choose $c=\frac{1}{2}(m+M)$. By Proposition \ref{LTT} one infers
that
\begin{equation}
\label{AAp}
\|Tx\|^2-\lal Tx,x\ral^2=\|T_cx\|^2-|\lal T_c x,x\ral|^2\quad \text{for any\, $x\in\fH$, $\|x\|=1$},
\end{equation}
where $T_c=T-cI$. Notice that $\|T_c\|=\frac{1}{2}(M-m)$. Then from \eqref{AAp} it follows that
\begin{align}
\|Tx\|^2-\lal Tx,x\ral^2&\leq\|T_cx\|^2\leq \frac{1}{4}(M-m)^2,
\end{align}
proving the right inequality in \eqref{AMm} and, thus, completing
the whole proof.
\end{proof}

\begin{remark}
Combining equalities \eqref{HPexnorm} and \eqref{HPnormEE} in Example \ref{exmpl} proves that the upper
bound in \eqref{AMm} is sharp.
\end{remark}

Our final bound only concerns the special case of bounded
Hamiltonians.

\begin{theorem}
\label{ThFin} Let $H$ be a bounded self-adjoint operator
on the Hilbert space $\fH$ and let
\begin{equation}
\label{Emima} E_{\rm
min}(H)=\min\bigl(\mathop{\mathrm{spec}}(H)\bigr),\quad E_{\rm
max}(H)=\max\bigl(\mathop{\mathrm{spec}}(H)\bigr)
\end{equation}
be the upper and lower bounds of the spectrum of $H$, respectively.
Let $P_t$, $t\geq 0$, be the projection path \eqref{PT} where $P_0$
is an orthogonal projection in $\mathfrak{H}$. Then
\begin{equation}
\label{FBound} \vartheta\bigl(\fP_s,\fP_t\bigr)\leq \frac{E_{\rm
max}(H)-E_{\rm min}(H)}{2}\, |t-s|,\quad \text{for any
\,}s,t\in\bbR,
\end{equation}
where $\fP_\tau= \Ran(P_\tau)$, $\tau\in\bbR$. Furthermore, the
following lower bound holds\footnote{For the case where $E_{\rm
max}(H)-E_{\rm min}(H)=0$, which implies that $H$ is a multiple of
unity, one adopts the convention that $T_\theta=\infty$ (cf. footnotes on
pages \pageref{pref1} and \pageref{pref2}).}:
\begin{equation}
\label{TtPbE}
T_\theta\,
\geq \frac{2\theta}{E_{\rm max}(H)-E_{\rm min}(H)}\,,
\end{equation}
where $\theta$ and $T_\theta$ are the same as in Corollary \ref{Cor1}.
\end{theorem}
\begin{proof}
In the case under consideration, the quantity $E_{\mathfrak{P}_0}$
from \eqref{Delta} may be rewritten as
\begin{align}
\label{DeltaB}
\Delta E_{\mathfrak{P}_0}=&\sup\limits_{\scriptsize\begin{matrix}f\in\fP_0\\
\|f\|=1\end{matrix}} \bigl(\|Hf\|^2-{\langle}
Hf,f{\rangle}^2\bigr)^{1/2}
\leq\sup\limits_{\scriptsize\begin{matrix}f\in\fH\\ \|f\|=1
\end{matrix}} \bigl(\|Hf\|^2-{\langle}
Hf,f{\rangle}^2\bigr)^{1/2},
\end{align}
since the self-adjoint operator $H$ is bounded and $\Dom(H)=\fH$.
Then by Lemma \ref{cMm} one concludes that
\begin{equation}
\label{EEmima} \Delta E_{\mathfrak{P}_0}\leq \frac{E_{\rm
max}(H)-E_{\rm min}(H)}{2},\quad \text{for any subspace \,}\fP_0\subset\fH.
\end{equation}
Now inequality \eqref{TBoundE} in Theorem \ref{TTtheta} implies the
estimate \eqref{FBound} while the lower bound \eqref{TtPbE} follows
from the lower bound \eqref{TtPb1}, and this completes the proof.
\end{proof}

\section{Summary and future perspectives}
\label{sConcl}

This paper is aimed in establishing sharp lower bounds on the time
required for an initial state subspace $\fP_0$ of a quantum system
described by a Hamiltonian $H$ to evolve into another subspace
particularly positioned with respect to the subspace $\fP_0$. The
operator $H$ is assumed to be time-independent and self-adjoint; it
is allowed to be unbounded. In the latter case, we require, in
addition, that the domain $\Dom(H)$ of $H$ is invariant under the
orthogonal projection $P_0$ on $\fP_0$, that is, $P_0 \Dom(H)\subset
\Dom(H)$. As a measure of the difference between the subspaces
$\fP_s$ and $\fP_t$, $s,t\in\bbR$, in the subspace evolution path
generated by $H$ out of $\fP_0$ we use the maximal angle
$\vartheta(\fP_s,\fP_t)$ between them, defined as in
\eqref{thetmax}.

Our first principal result (see Theorem \ref{Th2}) is as follows.
Assume that the commutator of $H$ and $P_0$ is a bounded operator on
$\Dom(H)$ and let $V_{H,P_0}=\|[H,P_0]\|$. Then  for any
$s,t\in\bbR$ the maximal angle between the subspaces $\fP_s$ and
$\fP_t$ satisfies inequality  $\vartheta(\fP_s,\fP_t)\leq
V_{H,P_0}|s-t|$. Furthermore, we show that this inequality is
optimal. Thus, it is the quantity $V_{H,P_0}$ that bounds the
maximal possible ``angular'' speed of variation of the subspaces in
the path $\fP_t$, $t\in\bbR$. In particular, we make the conclusion
that $T_\theta\geq \theta/V_{{H,P_0}}$ where $T_\theta$ is the first
time moment when the maximal angle between $\fP_0$ and $\fP_t$
reaches the value of $\theta$, $0<\theta\leq \frac{\pi}{2}$ (see
Corollary \ref{Cor1}).

The second principal result of the paper is based on the bound
\eqref{VHPbf2}, $V_{H,P_0}\leq \Delta E_{\mathfrak{P}_0}$, where
$\Delta E_{\mathfrak{P}_0}$ denotes the maximum energy dispersion
\eqref{Delta} over the subspace $\fP_0$. By using this bound we
derive for $T_\theta$ another sharp lower estimate
$T_\theta\geq \theta/ \Delta E_{\mathfrak{P}_0}$ (see Theorem
\ref{TTtheta} and Remark \ref{remref}) which is, in general, weaker
than our previous lower limit $T_\theta\geq \theta/V_{{H,P_0}}$.
However the estimate $T_\theta\geq \theta/ \Delta
E_{\mathfrak{P}_0}$ is featuring much closer resemblance to a
Mandelstam-Tamm-Fleming bound \eqref{Flem}, which is a particular
case of it.

Surely, it is also interesting to have  for $T_\theta$  a lower
bound written directly in terms of the spectrum of the Hamiltonian
$H$. We have found a bound of such a kind, namely the bound
\eqref{TtPbE} in Theorem \ref{ThFin}. It should be underlined,
however, that this bound only works in the case where
$H$ is a bounded operator.

We believe that the established bounds on the speed of evolution of
a quantum subspace admit an extension to the case of some unbounded
time-dependent self-adjoint Hamiltonians, and we work on this. It is
worth noting that this extension is not straightforward since,
unlike the one-parameter evolution semigroup entries \eqref{UU1},
the corresponding propagators do not commute with a time-dependent
Hamiltonian.

Another challenging problem consists in establishing bounds on the
evolution speed for a subspace whose variation is governed by a
non-Hermitian Hamiltonian. Of course, the general non-Hermitian case
should be extremely difficult. Furthermore, this case is not of
direct interest for quantum physics. However, special cases where a
non-Hermitian Hamiltonian is similar to a self-adjoint operator are
definitely of interest. This concerns, in particular, the case of
$PT$-symmetric Hamiltonians with real spectrum where some quantum
speed limits have already been established for the evolution of pure
states (see \cite{BBr}, \cite{China}, and references therein). One
may expect that in this case quantum speed limits exist for the
evolution of subspaces as well.


\end{document}